\documentclass{article}
\usepackage{dsfont}
\usepackage{graphicx}
\usepackage{psfrag}
\usepackage{stmaryrd}
\usepackage{subfigure}
\usepackage{amsmath}
\usepackage{alltt}
\usepackage{amsfonts}
\usepackage{amssymb}
\usepackage{mathrsfs}
\usepackage{txfonts}
\usepackage{booktabs}
\usepackage{times}
\usepackage{algorithm2e}
\usepackage{algorithmic}
\usepackage{tabularx}
\usepackage[standard]{ntheorem}
\usepackage{mathptmx}
\usepackage{helvet}
\usepackage{courier}
\usepackage{makeidx}
\usepackage{multicol}
\usepackage{footmisc}

\begin{document}
\title{A Topological Study of Chaotic Iterations\\Application to Hash Functions}

\author{Christophe Guyeux \and Jacques M. Bahi}

%\institute{Christophe Guyeux and Jacques M. Bahi \at Computer Science Laboratory LIFC, University of Franche-Comt\'e\\16, route de Gray - 25030 Besan\c con, France\\Phone: +33 381666948; email: \{christophe.guyeux, jacques.bahi\}@univ-fcomte.fr.}

\maketitle
\abstract{Chaotic iterations, a tool formerly used in distributed computing, has recently revealed various interesting properties of disorder leading to its use in the computer science security field.
In this paper, a comprehensive study of its topological behavior is proposed.
It is stated that, in addition to being chaotic as defined in the Devaney's formulation, this tool possesses the property of topological mixing.
Additionally, its level of sensibility, expansivity, and topological entropy are evaluated.
All of these properties lead to a complete unpredictable behavior for the chaotic iterations.
As it only manipulates binary digits or integers, we show that it is possible to use it to produce truly chaotic computer programs.
As an application example, a truly chaotic hash function is proposed in two versions.
In the second version, an artificial neural network is used, which can be stated as chaotic according to Devaney.}

%\begin{keywords}
%Topological chaos; Chaotic iterations; Hash functions; Neural Networks.
%\end{keywords}

\section{INTRODUCTION}
Chaotic iterations (CIs) were formerly a way to formalize distributed algorithms through mathematical tools \cite{Chazan69}.
By using these CIs, it was thus possible to study the convergence of synchronous or asynchronous programs over parallel, distributed, P2P, grid, or GPU platforms, in a view to solve linear and non-linear systems.
We have proven at the IEEE World Congress on Computational Intelligence (WCCI'10) that CIs can behave chaotically, as it is defined by Devaney \cite{guyeux10}.
These proofs have been improved and more detailed in \cite{guyeux09}.
In this paper, which is an extension of \cite{guyeux10,guyeux09}, we notably enlarge the theoretical study of CIs, among other things by computing its topological entropy, to obtain a comprehensive evaluation of its topological behavior.
This study leads us to the conclusion that the chaos of CIs is very intense and constitutes a useful tool to be used in the computer science security field.

Chaos in information security fields as digital watermarking \cite{CongJQZ06,Zhu06}, hash functions \cite{Wang2003,Xiao20094346}, or pseudo-random number generators, is often disputed.
This is due to the fact that this use is almost always based on the conception of algorithms that only include ``somewhere'' some well-known chaotic real functions like logistic, tent, or Arnold's cat maps, to obtain a program supposed to express these chaotic properties \cite{Fei2005,Wu2007bis,Zhou1997429}.
However, using such functions with other ``obvious'' parameters does not guarantee that the whole algorithm still remains chaotic.
Such an assumption should at least be discussed.
Moreover, even if the algorithm obtained by the inclusion of chaotic maps is itself proven to be chaotic, its implementation on a machine can lead to the fact that this chaotic nature is lost.
This is due to the finite cardinality of the machine numbers set.
%These issues are discussed in this document.

In this paper, as in \cite{guyeux10,guyeux09}, we do not simply integrate chaotic maps into algorithms hoping that the result remains chaotic.
We conceive algorithms for computer security that we have mathematically proven to be chaotic, as it is defined in the Devaney's theory.
We raise the question of their implementation, proving in doing so that it is possible to design both a chaotic algorithm and its associated chaotic computer program.
The chaos theory we consider is taken from the mathematical topology.
It encompasses the well-known Devaney's definition of chaos and the notions of expansivity, topological entropy, and topological mixing.
These notions of unpredictability are the most established ones into the mathematical theory of chaos.
Our fundamental study is motivated by the desire to produce chaotic programs in the area of information security.

The paper begins by introducing the theoretical foundation of this approach.
On the one hand we recall the definition of Devaney's topological chaos and on the other hand the definition of discrete chaotic iterations.
Although these definitions are distinct from each other, we establish a link between them by giving conditions under which chaotic discrete iterations generate a Devaney's topological chaos.
This study is deepened by giving some qualitative and quantitative evaluations of the disorder generated by chaotic iterations.
These evaluations are not present in \cite{guyeux09,guyeux10}.
We will focus in this paper on the notions of expansivity, topological mixing, and topological entropy.
The proofs that the considered space is separated and compact have never been published before.
Furthermore, the evaluation of the topological entropy is completely new.
Then, because chaotic iterations are very suited for computer programming, this link allows us to generate programs in the computer science field that behave chaotically.
This link was formerly presented in \cite{guyeux10,guyeux09} with some errors that are corrected here.

After having studied the theoretical aspects of our approach we focus on the practical ones.
The important question is how to preserve the topological chaos properties in a set of a finite number of states.
This question is answered in Section \ref{section:CHAOS IN A FINITE STATE MACHINE}, by manipulating only integers and considering the use of new data at each iteration.
%Finally, applications in the domain of hash functions are given.
The general algorithm based on our approach, formerly presented in \cite{guyeux10}, is now explained in detail in Section \ref{section:HASH FUNCTIONS BASED ON TOPOLOGICAL CHAOS}.
It is rewritten in the next section as an artificial neural network that can compute hash values while behaving chaotically.
This new application is completely new and has never been published before.

\medskip

The remainder of this paper is organized in the following way.
In Section \ref{section:BASIC RECALLS}, the definitions of Devaney's chaos and discrete chaotic iterations are recalled.
A link between these two notions is established and sufficient conditions to obtain Devaney's topological chaos from discrete chaotic iterations are given in Section \ref{section:CHAOTIC ITERATIONS AS DEVANEY'S CHAOS}.
The chaotic behavior of CIs is deepened in Section \ref{section:PROPERTIES OF THE CHAOTIC MACHINE}, by studying some qualitative and quantitative properties of disorder.
In Section \ref{section:CHAOS IN A FINITE STATE MACHINE}, the question on how to preserve these chaotic properties into computers is answered.
Then in Section \ref{section:HASH FUNCTIONS BASED ON TOPOLOGICAL CHAOS} the general hash function scheme is given and illustrated, whereas in Section \ref{sec:A CHAOTIC NEURAL NETWORK AS HASH FUNCTION} it is applied to produce an artificial neural network able to hash some values in a chaotic manner.
The paper ends by a conclusion section in which our contribution is summarized and planned future work is discussed.

% >>>>>>>>>>>>>>>>>>>>>> Basic recalls <<<<<<<<<<<<<<<<<<<<<<<<<<<<<<
\section{BASIC RECALLS}
\label{section:BASIC RECALLS}
This section is devoted to basic definitions and terminologies in the field of topological chaos and in the one of chaotic iterations.
\subsection{Devaney's chaotic dynamical systems}

In the sequel $S^{n}$ denotes the $n^{th}$ term of a sequence $S$ and $V_{i}$ denotes the $i^{th}$ component of a vector $V$. $f^{k}=f\circ ...\circ f$ denotes the $k^{th}$ composition of a function $f$. Finally, the following notation is used: $\llbracket1;N\rrbracket=\{1,2,\hdots,N\}$.

Consider a topological space $(\mathcal{X},\tau)$ and a continuous function $f : \mathcal{X} \rightarrow \mathcal{X}$.

\begin{definition}
$f$ is said to be \emph{topologically transitive} if, for any pair of open sets $U,V \subset \mathcal{X}$, there exists $k>0$ such that $f^k(U) \cap V \neq \varnothing$.
\end{definition}

\begin{definition}
An element $x$ is a \emph{periodic point} for $f$ of period $n\in \mathds{N}^*$ if $f^{n}(x)=x$.% The set of periodic points of $f$ is denoted $Per(f).$
\end{definition}

\begin{definition}
$f$ is said to be \emph{regular} on $(\mathcal{X}, \tau)$ if the set of periodic points for $f$ is dense in $\mathcal{X}$: for any point $x$ in $\mathcal{X}$, any neighborhood of $x$ contains at least one periodic point (without necessarily the same period).
\end{definition}

\begin{definition}
$f$ is said to be \emph{chaotic} on $(\mathcal{X},\tau)$ if $f$ is regular and topologically transitive.
\end{definition}

The chaos property is strongly linked to the notion of ``sensitivity'', defined on a metric space $(\mathcal{X},d)$ by:

\begin{definition}
\label{sensitivity} $f$ has \emph{sensitive dependence on initial conditions}
if there exists $\delta >0$ such that, for any $x\in \mathcal{X}$ and any neighborhood $V$ of $x$, there exist $y\in V$ and $n > 0$ such that $d\left(f^{n}(x), f^{n}(y)\right) >\delta $.

$\delta$ is called the \emph{constant of sensitivity} of $f$.
\end{definition}

Indeed, Banks \emph{et al.} have proven in~\cite{Banks92} that when $f$ is chaotic and $(\mathcal{X}, d)$ is a metric space, then $f$ has the property of sensitive dependence on initial conditions (this property was formerly an element of the definition of chaos). To sum up, quoting Devaney in~\cite{Devaney}, a chaotic dynamical system ``is unpredictable because of the sensitive dependence on initial conditions. It cannot be broken down or simplified into two subsystems which do not interact because of topological transitivity. And in the midst of this random behavior, we nevertheless have an element of regularity''. Fundamentally different behaviors are consequently possible and occur in an unpredictable way.

\subsection{Chaotic iterations}
\label{sec:chaotic iterations}

Let us consider  a \emph{system} with a finite  number $\mathsf{N} \in
\mathds{N}^*$ of elements  (or \emph{cells}), so that each  cell has a
Boolean  \emph{state}. Having $\mathsf{N}$ Boolean values for these
 cells  leads to the definition of a particular \emph{state  of the
system}. A sequence which  elements belong to $\llbracket 1;\mathsf{N}
\rrbracket $ is called a \emph{strategy}. The set of all strategies is
denoted by $\mathbb{S}.$

\begin{definition}
\label{Def:chaotic iterations}
The      set       $\mathds{B}$      denoting      $\{0,1\}$,      let
$f:\mathds{B}^{\mathsf{N}}\longrightarrow  \mathds{B}^{\mathsf{N}}$ be
a  function  and  $S\in  \mathbb{S}$  be  a  strategy.  The  so-called
\emph{chaotic      iterations}     are     defined      by     $x^0\in
\mathds{B}^{\mathsf{N}}$ and
$$
\forall    n\in     \mathds{N}^{\ast     },    \forall     i\in
\llbracket1;\mathsf{N}\rrbracket ,x_i^n=\left\{
\begin{array}{ll}
  x_i^{n-1} &  \text{ if  }S^n\neq i \\
  \left(f(x^{n-1})\right)_{S^n} & \text{ if }S^n=i.
\end{array}\right.
$$
\end{definition}

In other words, at the $n^{th}$ iteration, only the $S^{n}-$th cell is
\textquotedblleft  iterated\textquotedblright .  Note  that in  a more
general  formulation,  $S^n$  can   be  a  subset  of  components  and
$\left(f(x^{n-1})\right)_{S^{n}}$      can     be      replaced     by
$\left(f(x^{k})\right)_{S^{n}}$, where  $k<n$, describing for example,
delays  transmission~\cite{Robert1986,GuyeuxThese10}.  Finally,  let us  remark that
the term  ``chaotic'', in  the name of  these iterations,  has \emph{a
priori} no link with the mathematical theory of chaos, recalled above.

\section{CHAOTIC ITERATIONS AS DEVANEY'S CHAOS}
\label{section:CHAOTIC ITERATIONS AS DEVANEY'S CHAOS}

In this section is proven that chaotic iterations are a particular case of topological chaos, as it is defined in the Devaney's formulation'.

\subsection{The new topological space}

In this section we define a suitable metric space where chaotic iterations are continuous.
\subsubsection{Defining the iteration function and the phase space}
\label{Defining}
Let $\delta $ be the \emph{discrete Boolean metric}, $\delta (x,y)=0\Leftrightarrow x=y.$ Given a function $f$, define the function:
\begin{equation*}
\begin{array}{lrll}
F_{f}: & \llbracket1;\mathsf{N}\rrbracket\times \mathds{B}^{\mathsf{N}} &
\longrightarrow & \mathds{B}^{\mathsf{N}} \\
& (k,E) & \longmapsto & \left( E_{j}.\delta (k,j)+f(E)_{k}.\overline{\delta
(k,j)}\right) _{j\in \llbracket1;\mathsf{N}\rrbracket},%
\end{array}%
\end{equation*}%
\noindent where + and . are the Boolean addition and product operations.
Consider the phase space:
\begin{equation*}
\mathcal{X} = \llbracket 1 ; \mathsf{N} \rrbracket^\mathds{N} \times
\mathds{B}^\mathsf{N},
\end{equation*}
\noindent and the map defined on $\mathcal{X}$:
\begin{equation}
G_f\left(S,E\right) = \left(\sigma(S), F_f(i(S),E)\right), \label{Gf}
\end{equation}
\noindent where $\sigma$ is the \emph{shift} function defined by $\sigma (S^{n})_{n\in \mathds{N}}\in \mathbb{S}\longrightarrow (S^{n+1})_{n\in \mathds{N}}\in \mathbb{S}$ and $i$ is the \emph{initial function}  $i:(S^{n})_{n\in \mathds{N}} \in \mathbb{S}\longrightarrow S^{0}\in \llbracket 1;\mathsf{N}\rrbracket$. Then the chaotic iterations defined in (\ref{sec:chaotic iterations}) can be described by the following iterations:
\begin{equation*}
\left\{
\begin{array}{l}
X^0 \in \mathcal{X} \\
X^{k+1}=G_{f}(X^k).%
\end{array}%
\right.
\end{equation*}%

With this formulation, a shift function appears as a component of chaotic iterations. The shift function is a famous example of a chaotic map~\cite{Devaney} but its presence is not sufficient enough to claim $G_f$ as chaotic. In the rest of this section we rigorously prove that under some hypotheses, chaotic iterations generate topological chaos. Furthermore, due to the suitability of chaotic iterations for computer programming \cite{bgw09:ip,wbg10:ip,guyeux10ter} we also have proven that this is true in the computer science field.

\subsubsection{Cardinality of $\mathcal{X}$}
By comparing $\mathbb{S}$ and $\mathds{R}$, we have the following result.
\begin{theorem}
The phase space $\mathcal{X}$ has, at least, the cardinality of the continuum.
\end{theorem}
\begin{proof}
Let $\varphi$ be the map which transforms a strategy into the binary representation of an element in $[0,1$[, as follows. If the $n^{th}$ term of the strategy is 0, then the $n^{th}$ associated digit is 0. If this $n^{th}$ term is not equal to 0, then the associated digit is 1.
With this construction, $\varphi : \llbracket 1 ; \mathsf{N} \rrbracket^\mathds{N} \longrightarrow [0,1]$ is onto. But $]0,1[$ is isomorphic to $\mathds{R}$ ($x \in ]0,1[\mapsto tan(\pi(x-\frac{1}{2}))$ is an isomorphism), so the cardinality of $\llbracket 1 ; \mathsf{N} \rrbracket^\mathds{N}$ is greater or equal to the cardinality of $\mathds{R}$. As a consequence, the cardinality of the Cartesian product $\mathcal{X} = \llbracket 1 ; \mathsf{N} \rrbracket^\mathds{N} \times \mathds{B}^\mathsf{N}$ is greater or equal to the cardinality of $\mathds{R}$.
\end{proof}
\begin{remark}
This result is independent from the number of components of the system.
\end{remark}

\subsubsection{A new distance}

Let us define a new distance between two points $X = (S,E), Y = (\check{S},\check{E})\in
\mathcal{X}$ by%
\begin{equation*}
d(X,Y)=d_{e}(E,\check{E})+d_{s}(S,\check{S}),
\end{equation*}
\noindent where
\begin{equation*}
\left\{
\begin{array}{lll}
\displaystyle{d_{e}(E,\check{E})} & = & \displaystyle{\sum_{k=1}^{\mathsf{N}%
}\delta (E_{k},\check{E}_{k})}, \\
\displaystyle{d_{s}(S,\check{S})} & = & \displaystyle{\dfrac{9}{\mathsf{N}}%
\sum_{k=1}^{\infty }\dfrac{|S^k-\check{S}^k|}{10^{k}}}.%
\end{array}%
\right.
\end{equation*}

This new distance has been introduced in \cite{guyeux10} to satisfy the following requirements.
\begin{itemize}
\item When the number of different cells between two systems is increasing, then their distance should increase too.
\item In addition, if two systems present the same cells and their respective strategies start with the same terms, then the distance between these two points must be small because the evolution of the two systems will be the same for a while. Indeed, the two dynamical systems start with the same initial condition, use the same update function, and as strategies are the same for a while, then components that are updated are the same too.
\end{itemize}
The distance presented above follows these recommendations. Indeed, if the floor value $\lfloor d(X,Y)\rfloor $ is equal to $n$, then the systems $E, \check{E}$ differ in $n$ cells. In addition, $d(X,Y) - \lfloor d(X,Y) \rfloor $ is a measure of the differences between strategies $S$ and $\check{S}$. More precisely, this floating part is less than $10^{-k}$ if and only if the first $k$ terms of the two strategies are equal. Moreover, if the $k^{th}$ digit is nonzero, then the $k^{th}$ terms of the two strategies are different.

\subsubsection{Continuity of the iteration function}

To prove that chaotic iterations are an example of topological chaos in the
sense of Devaney ~\cite{Devaney}, $G_{f}$ must be continuous in the metric
space $(\mathcal{X},d)$.

\begin{theorem}
$G_f$ is a continuous function.
\end{theorem}

\begin{proof}
We use the sequential continuity.
Let $(S^n,E^n)_{n\in \mathds{N}}$ be a sequence of the phase space $%
\mathcal{X}$, which converges to $(S,E)$. We will prove that $\left(
G_{f}(S^n,E^n)\right) _{n\in \mathds{N}}$ converges to $\left(
G_{f}(S,E)\right) $. Let us recall that for all $n$, $S^n$ is a strategy,
thus, we consider a sequence of strategies (\emph{i.e.}, a sequence of
sequences).\newline
As $d((S^n,E^n);(S,E))$ converges to 0, each distance $d_{e}(E^n,E)$ and $d_{s}(S^n,S)$ converges
to 0. But $d_{e}(E^n,E)$ is an integer, so $\exists n_{0}\in \mathds{N},$ $%
d_{e}(E^n,E)=0$ for any $n\geqslant n_{0}$.\newline
In other words, there exists a threshold $n_{0}\in \mathds{N}$ after which no
cell will change its state:
$\exists n_{0}\in \mathds{N},n\geqslant n_{0}\Rightarrow E^n = E.$

In addition, $d_{s}(S^n,S)\longrightarrow 0,$ so $\exists n_{1}\in %
\mathds{N},d_{s}(S^n,S)<10^{-1}$ for all indexes greater than or equal to $%
n_{1}$. This means that for $n\geqslant n_{1}$, all the $S^n$ have the same
first term, which is $S^0$: $\forall n\geqslant n_{1},S_0^n=S_0.$

Thus, after the $max(n_{0},n_{1})^{th}$ term, states of $E^n$ and $E$ are
identical and strategies $S^n$ and $S$ start with the same first term.\newline
Consequently, states of $G_{f}(S^n,E^n)$ and $G_{f}(S,E)$ are equal,
so, after the $max(n_0, n_1)^{th}$ term, the distance $d$ between these two points is strictly less than 1.\newline
\noindent We now prove that the distance between $\left(
G_{f}(S^n,E^n)\right) $ and $\left( G_{f}(S,E)\right) $ is convergent to
0. Let $\varepsilon >0$. \medskip
\begin{itemize}
\item If $\varepsilon \geqslant 1$, we see that distance
between $\left( G_{f}(S^n,E^n)\right) $ and $\left( G_{f}(S,E)\right) $ is
strictly less than 1 after the $max(n_{0},n_{1})^{th}$ term (same state).
\medskip
\item If $\varepsilon <1$, then $\exists k\in \mathds{N},10^{-k}\geqslant
\varepsilon > 10^{-(k+1)}$. But $d_{s}(S^n,S)$ converges to 0, so
\begin{equation*}
\exists n_{2}\in \mathds{N},\forall n\geqslant
n_{2},d_{s}(S^n,S)<10^{-(k+2)},
\end{equation*}%
thus after $n_{2}$, the $k+2$ first terms of $S^n$ and $S$ are equal.
\end{itemize}
\noindent As a consequence, the $k+1$ first entries of the strategies of $%
G_{f}(S^n,E^n)$ and $G_{f}(S,E)$ are the same ($G_{f}$ is a shift of strategies) and due to the definition of $d_{s}$, the floating part of
the distance between $(S^n,E^n)$ and $(S,E)$ is strictly less than $%
10^{-(k+1)}\leqslant \varepsilon $.\bigskip \newline
In conclusion,
$$
\forall \varepsilon >0,\exists N_{0}=max(n_{0},n_{1},n_{2})\in \mathds{N}%
,\forall n\geqslant N_{0},
 d\left( G_{f}(S^n,E^n);G_{f}(S,E)\right)
\leqslant \varepsilon .
$$
$G_{f}$ is consequently continuous.
\end{proof}

In this section, we proved that chaotic iterations can be modeled as a dynamical system in a topological space. In the next section, we show that some chaotic iterations behave chaotically, as defined by Devaney's theory.
% >>>>>>>>>>>>>>>>>>>>>> Discrete chaotic iterations as topological chaos <<<<<<<<<<<<<<<<<<<<<<<<<<<<<<

\subsection{Discrete chaotic iterations as topological chaos}

To prove that we are in the framework of Devaney's topological chaos, we have to find a Boolean function $f$ such that $G_f$ satisfies the regularity, transitivity, and sensitivity conditions.
We will prove that the vectorial logical negation

\begin{equation}
f_{0}(x_{1},%
\hdots,x_{\mathsf{N}})=(\overline{x_{1}},\hdots,\overline{x_{\mathsf{N}}})
\label{f0}
\end{equation}
\noindent is a suitable function.

\subsubsection{Regularity}
\label{regularite}

Firstly, let us prove that,

\begin{theorem}
Periodic points of $G_{f_0}$ are dense in $\mathcal{X}$.
\end{theorem}

\begin{proof}
Let $(\check{S}, \check{E})\in \mathcal{X}$ and $\varepsilon >0$. We are
looking for a periodic point $(\widetilde{S},\widetilde{E})$ satisfying $d((%
\check{S}, \check{E});(\widetilde{S},\widetilde{E}))<\varepsilon$.
As $\varepsilon$ can be strictly lesser than 1, we must choose $%
\widetilde{E} = \check{E}$. Let us define $k_0(\varepsilon) =\lfloor
log_{10}(\varepsilon )\rfloor +1$ and consider the set
\[
\mathcal{S}_{\check{S}, k_0(\varepsilon)} = \left\{ S \in \mathbb{S} / S^k =
\check{S}^k, \forall k \leqslant k_0(\varepsilon) \right\}.
\]
Then, $\forall S \in \mathcal{S}_{\check{S}, k_0(\varepsilon)}, d((S, \check{%
E});(\check{S}, \check{E})) < \varepsilon$. It remains to choose $\widetilde{%
S} \in \mathcal{S}_{\check{S}, k_0(\varepsilon)}$ such that $(\widetilde{S},%
\widetilde{E}) = (\widetilde{S},\check{E})$ is a periodic point for $%
G_{f_0}$.
Let

\begin{center}
$\mathcal{J} = \left\{ i \in \{1, ..., \mathsf{N}\} / E_i \neq \check{%
E}_i, \text{ where } (S, E) = G_{f_0}^{k_0} (\check{S}, \check{E}) \right\}$%
,
\end{center}

\noindent $i_0 = card(\mathcal{J})$, and $j_1 <j_2 < ... < j_{i_0}$ the elements of $%
\mathcal{J}$. Then, $\widetilde{S} \in \mathcal{S}_{\check{S},
k_0(\varepsilon)}$ defined by
\begin{itemize}
\item $\widetilde{S}^k = \check{S}^k$, if $k \leqslant k_0(\varepsilon)$,
\item $\widetilde{S}^k = j_{k-k_0(\varepsilon)}$, if $k \in
\{k_0(\varepsilon)+1, k_0(\varepsilon)+2, ..., k_0(\varepsilon)+i_0\}$,
\item and $\widetilde{S}^{k}=\widetilde{S}^{j}$, where $j\leqslant
k_{0}(\varepsilon )+i_{0}$ is satisfying $j\equiv k~(\text{mod }%
k_{0}(\varepsilon )+i_{0})$, if $k>k_{0}(\varepsilon )+i_{0}$,
\end{itemize}
\noindent is such that $%
(\widetilde{S},\widetilde{E})$ is a periodic point (of period $%
k_{0}(\varepsilon )+i_{0}$), which is $\varepsilon -$close to $(\check{S},%
\check{E})$.\newline As a conclusion, $(\mathcal{X},G_{f_{0}})$ is
regular.
\end{proof}

\subsubsection{Transitivity}

Regarding the transitivity property of $G_{f_0}$, we can show that,

\begin{theorem}
$(\mathcal{X},G_{f_0})$ is topologically transitive.
\end{theorem}

\begin{proof}
Let us define $\mathcal{E}:\mathcal{X}\rightarrow \mathbb{B}^{\mathsf{N}},$
such that $\mathcal{E(}S,E)=E.$ Let $\mathcal{B}_{A}=\mathcal{B}(X_{A},r_{A})
$ and $\mathcal{B}_{B}=\mathcal{B}(X_{B},r_{B})$ be two open balls of $%
\mathcal{X}$, with $X_{A}=(S_{A},E_{A})$ and $X_{B}=(S_{B},E_{B})$. We are
looking for $\widetilde{X}=(\widetilde{S},\widetilde{E})$ in $\mathcal{B}_{A}
$ such that $\exists n_{0}\in \mathbb{N},G_{f_{0}}^{n_{0}}(\widetilde{X})\in
\mathcal{B}_{B}$.\newline
$\widetilde{X}$ must be in $\mathcal{B}_{A}$ and $r_{A}$ can be strictly
lesser than 1, so $\widetilde{E}=E_{A}$. Let $k_{0}=\lfloor \log
_{10}(r_{A})+1\rfloor $. Then $\forall S\in \mathbb{S}$, if $%
S^{k}=S_{A}^{k},\forall k\leqslant k_{0}$, then $(S,\widetilde{E})\in
\mathcal{B}_{A}$. Let $(\check{S},\check{E})$ be equal to $G_{f_{0}}^{k_{0}}(S_{A},E_{A})$ and $c_{1},...,c_{k_{1}}$ denote the elements of
the set $\{i\in \llbracket1,\mathsf{N}\rrbracket/\check{E}_{i}\neq \mathcal{E%
}(X_{B})_{i}\}.$ So any point $X$ of the set
\[
\{(S,E_{A})\in \mathcal{X}/\forall k\leqslant k_{0},S^{k}=S_{A}^{k}\text{
and }\forall k\in \llbracket1,k_{1}\rrbracket,S^{k_{0}+k}=c_{k}\}
\]%
is satisfying $X\in \mathcal{B}_{A}$ and $\mathcal{E}\left(
G_{f_{0}}^{k_{0}+k_{1}}(X)\right) =E_{B}$.
\noindent Lastly, let $k_2$ be \linebreak $\lfloor \log_{10}(r_B)\rfloor +1$. Then $%
\widetilde{X} = (\widetilde{S}, \widetilde{E}) \in \mathcal{X}$ defined by:
\begin{enumerate}
\item $\widetilde{X} = E_A$,
\item $\forall k \leqslant k_0, \widetilde{S}^k = S_A^k$,
\item $\forall k \in \llbracket 1, k_1 \rrbracket,$ $\widetilde{S}^{k_0+k} =
c_k$,
\item $\forall k \in \mathbb{N}^*, \widetilde{S}^{k_0+k_1+k} = S_B^k$,
\end{enumerate}
\noindent is such that $\widetilde{X} \in \mathcal{B}_A$ and $G_{f_0}^{k_0+k_1}(%
\widetilde{X}) \in \mathcal{B}_B$. This fact concludes the proof of the theorem.
\end{proof}

\subsubsection{Devaney's Chaos}
\label{sec:DevaneysChaos}

In conclusion, $(\mathcal{X},G_{f_0})$ is topologically transitive and regular. Then we have the following result:

\begin{theorem}
\label{theorem:Chaos}
$G_{f_0}$ is a chaotic map on $(\mathcal{X},d)$ in the sense of Devaney.
\end{theorem}

We have proven that the set $\mathcal{C}$ of the iterate functions $f$ so that $(\mathcal{X}, G_f)$ is chaotic (according to the definition of Devaney), is a nonempty set. In future work, we will deepen the study of $\mathcal{C}$, among other things, by computing its cardinality and characterizing this set.

\section{TOPOLOGICAL PROPERTIES OF CHAOTIC ITERATIONS}
\label{section:PROPERTIES OF THE CHAOTIC MACHINE}

In this section, some qualitative and quantitative topological properties for chaotic iterations with $G_{f_0}$ will be studied in detail. These properties reinforce the chaotic behavior of the system.

\subsection{Topological mixing}

The topological mixing is a strong version of transitivity:

\begin{definition}
A discrete dynamical system is said to be \emph{topologically mixing} if and only if, for any couple of disjoint open set $U, V \neq \varnothing$, $n_0 \in \mathbb{N}$ can be found so that $\forall n \geqslant n_0, f^n(U) \cap V \neq \varnothing$.
\end{definition}

We have the result \cite{gfb10:ip},

\begin{theorem}
$(\mathcal{X},G_{f_0})$ is topologically mixing.
\end{theorem}

This result is an immediate consequence of the lemma below.

\begin{lemma}
For any open ball $B$ of $\mathcal{X}$, an index $n$ can be found such that $G_{f_0}^n(B) = \mathcal{X}$.
\end{lemma}

\begin{proof}
Let $B=B((E,S),\varepsilon)$ be an open ball, which the radius can be considered as strictly less than 1.
The elements of $B$ all have the same state $E$ and are such that an integer $k \left(=-\log_{10}(\varepsilon)\right)$ satisfies:
\begin{itemize}
\item all the strategies of $B$ have the same $k$ first terms,
\item after the index $k$, all values are possible.
\end{itemize}

Then, after $k$ iterations, the new state of the system is $G_{f_0}^k(E,S)_1$ and all the strategies are possibles (any point of the form $(G_{f_0}^k(E,S)_1,\textrm{\^{S}})$, with any $\textrm{\^{S}} \in \mathbb{S}$, is reachable from $B$).

Let $(E',S') \in \mathcal{X}$. We will prove that any point of $\mathcal{X}$ is reachable from $B$.

Indeed, let $s_i$ be the list of the different cells between $G_{f_0}^k(E,S)_1$ and $E'$, and $|s|$ its size. The point $(\check{E},\check{S})$ of $B$ defined by:
\begin{itemize}
\item $\check{E} = E$,
\item $\check{S}^i = S^i, \forall i \leqslant k$,
\item $\check{S}^{k+i} = s_i, \forall i \leqslant |s|$,
\item $\forall i \in \mathds{N}, S^{k + |s| + i} =  S'^i$.
\end{itemize}
is such that $G_{f_0}^{k+|s|}(\check{E},\check{S}) = (E',S')$. This conclude the proof of the lemma.
\end{proof}

\subsection{Quantitative measures}
\label{QUANTITATIVE MEASURE}

\subsubsection{General definitions}

In Section \ref{sec:DevaneysChaos} we have proven that discrete chaotic iterations produce a topological chaos by checking two qualitative properties, namely transitivity and regularity. This mathematical framework offers tools to measure this chaos quantitatively.\newline
The first of these measures is the constant of sensitivity defined in Definition \ref{sensitivity}. Intuitively, a function $f$ having a constant  sensitivity equal to $\delta $ implies that there exists points arbitrarily close to any point $x$, which eventually separate from $x$ by at least $\delta $ under some iterations of $f$. This induces that an arbitrarily small error on an initial condition \emph{might be} magnified upon iterations of $f$.

Another important tool is defined below.

\begin{definition}
A function $f$ is said to have the property of \emph{expansivity} if
\begin{equation*}
\exists \varepsilon >0,\forall x\neq y,\exists n\in \mathbb{N}%
,d(f^{n}(x),f^{n}(y))\geqslant \varepsilon .
\end{equation*}

\noindent Then, $\varepsilon $ is the \emph{constant of expansivity }of $f.$ We also say that $f$ is $\varepsilon$-expansive.
\end{definition}

A function $f$ has a constant of expansivity equal to $\varepsilon $ if an arbitrarily small error on any initial condition is \emph{always} magnified until $\varepsilon $.

\subsubsection{Sensitivity}

\label{par:Sensitivity}
The sensitive dependence on the initial conditions has been shown as a consequence of the regularity and the transitivity of chaotic iterations. However, in the set of machine numbers, we have shown in \cite{guyeux10} that the notion of regularity must be redefined. This is the reason why this sensitivity should be proven without using the result of Banks \cite{Banks92}, to be sure that this dependence is preserved in practical use of chaotic iterations.

In addition, the constant of sensitivity will be obtained during this proof.

\begin{theorem}
$(\mathcal{X},G_{f_0})$ has sensitive dependence on initial conditions and its constant of sensitivity is equal to $\mathsf{N}-1$.
\end{theorem}

\begin{proof}
Let $\check{X} = (\check{S},\check{E}) \in \mathcal{X}$. We are looking for $\widetilde{X} = (\widetilde{S}, \widetilde{E}) \in \mathcal{X}$ such that $d(\check{X}, \widetilde{X}) \leqslant \delta$ and $\exists n_0 \in \mathbb{N}$, $d\left( G_{f_0}^{n_0} (\check{X}) ; G_{f_0}^{n_0} (\widetilde{X}) \right) \geqslant \mathsf{N}-1$. Let $k_{0}$ be $\lfloor \log_{10}(\delta) \rfloor +1$. So, if $S \in \{ S \in \mathbb{S}/ \forall k \leqslant k_0, S^k = \check{S}^k\}$, then $d\left( (S, \check{E}), (\check{S}, \check{E})\right) \leqslant \delta$.

Let $\mathcal{J} = \left\{ i \in \llbracket 1, \mathsf{N} \rrbracket / ~~\mathcal{E} \left( G_{f_0}^{k_0}(\check{S},\check{E}) \right)_i =  \mathcal{E} \left( G_{f_0}^{k_0+\mathsf{N}}(\check{S},\check{E}) \right)_i \right\}$ and $p = card \left(\mathcal{J}\right)$. If $p = \mathsf{N}$, then $(\widetilde{S}, \widetilde{E}) \in \mathcal{X}$ defined by:
\begin{enumerate}
\item $\widetilde{E} = \check{E}$,
\item $\forall k \leqslant k_0, \widetilde{S}^k = \check{S}^k$,
\item $\forall k \in \llbracket 1, \mathsf{N} \rrbracket, \widetilde{S}^{k_0+k} = k$,
\item $\forall k > k_0 + \mathsf{N}, \widetilde{S}^k = 1$.
\end{enumerate}
satisfies $d((\widetilde{S}, \widetilde{E}) ; (\check{S},\check{E})) < \delta$ and $\forall i \in \llbracket 1, \mathsf{N} \rrbracket, \mathcal{E} \left( G_{f_0}^{k_0+\mathsf{N}} (\widetilde{S} ; \widetilde{E}) \right)_i \neq \mathcal{E} \left( G_{f_0}^{k_0+\mathsf{N}} (\check{S} ; \check{E}) \right)_i$, so the result is obtained.\newline
Else, let $j_1 < j_2 < ... < j_p$ be the elements of $\mathcal{J}$ and $j_0 \notin \mathcal{J}$. Then $\widetilde{X} = (\widetilde{E}, \widetilde{S}) \in \mathcal{X}$ defined by
\begin{enumerate}
\item $\widetilde{E} = \check{E}$,
\item $\forall k \leqslant k_0, \widetilde{S}^k = \check{S}^k$,
\item $\forall k \in \llbracket 1, p \rrbracket, \widetilde{S}^{k_0+k} = j_k$,
\item $\forall k \in \mathbb{N}^*, \widetilde{S}^{k_0+p+k} = j_0$.
\end{enumerate}
is such that $d(\check{X}, \widetilde{X}) < \delta$. In addition, $\forall i \in \llbracket 1, p \rrbracket, \mathcal{E} \left(G_{f_0}^{k_0+\mathsf{N}} (\check{X})\right)_{j_i} \neq \mathcal{E} \left(G_{f_0}^{k_0+\mathsf{N}} (\widetilde{X})\right)_{j_i}$, because:
\begin{itemize}
\item $\forall i \in \llbracket 1, \mathsf{N} \rrbracket, \mathcal{E} \left(G_{f_0}^{k_0} (\check{X})\right)_{i} = \mathcal{E} \left(G_{f_0}^{k_0} (\widetilde{X})\right)_{i}$, due to the definition of $k_0$.
\item $\forall i \in \llbracket 1, p\rrbracket, j_i \in \mathcal{J} \Rightarrow \mathcal{E} \left(G_{f_0}^{k_0+\mathsf{N}} (\check{X})\right)_{j_i} =  \mathcal{E} \left(G_{f_0}^{k_0} (\check{X})\right)_{j_i}$, according to the definition of $\mathcal{J}$.
\item $\forall i \in \llbracket 1, p\rrbracket, j_i $ appears exactly one time in $\widetilde{S}^{k_0}, \widetilde{S}^{k_0+1}, ..., \widetilde{S}^{k_0+\mathsf{N}}$, so $$\mathcal{E} \left( G_{f_0}^{k_0+\mathsf{N}} (\widetilde{X})\right)_{j_i} \neq  \mathcal{E} \left( G_{f_0}^{k_0} (\widetilde{X})\right)_{j_i}.$$
\end{itemize}
Lastly, $\forall i \in \llbracket 1, \mathsf{N} \rrbracket \setminus \{j_0, j_1, ..., j_p \}, \mathcal{E} \left(G_{f_0}^{k_0+\mathsf{N}} (\widetilde{X})\right)_i \neq  \mathcal{E} \left(G_{f_0}^{k_0+\mathsf{N}} (\check{X})\right)_i$, because:
\begin{itemize}
\item $\forall i \in \llbracket 1, \mathsf{N} \rrbracket, \mathcal{E} \left(G_{f_0}^{k_0} (\check{X})\right)_{i} = \mathcal{E} \left(G_{f_0}^{k_0} (\widetilde{X})\right)_{i}$,
\item $i \notin \mathcal{J} \Rightarrow \mathcal{E} \left(G_{f_0}^{k_0+\mathsf{N}} (\check{X})\right)_{i} \neq \mathcal{E} \left(G_{f_0}^{k_0} (\check{X})\right)_{i}$,
\item $i \notin \{\widetilde{S}^{k_0}, \widetilde{S}^{k_0+1}, ..., \widetilde{S}^{k_0+\mathsf{N}}\} \Rightarrow \mathcal{E} \left(G_{f_0}^{k_0+\mathsf{N}} (\widetilde{X})\right)_{i} = \mathcal{E} \left(G_{f_0}^{k_0} (\widetilde{X})\right)_{i}$.
\end{itemize}
So, in this case, $\forall i \in \llbracket 1, \mathsf{N} \rrbracket \setminus \{j_0 \}, \mathcal{E} \left( G_{f_0}^{k_0+\mathsf{N}} (\widetilde{S} ; \widetilde{E}) \right)_i \neq \mathcal{E} \left( G_{f_0}^{k_0+\mathsf{N}} (\check{S} ; \check{E}) \right)_i$ and the result of sensitivity is still obtained.
\end{proof}

\subsubsection{Expansivity}

In this section we offer the proof that chaotic iterations are expansive \cite{gfb10:ip} when $f_0$ is the update function:

\begin{theorem}
$(\mathcal{X},G_{f_{0}})$ is an expansive chaotic system. Its constant of
expansivity is equal to 1.
\end{theorem}

\begin{proof}
If $(S,E)\neq (\check{S};\check{E})$, then:

\begin{itemize}
\item Either $E\neq \check{E}$ and so at least one cell is not in the
same state in $E$ and $\check{E}$. Consequently the distance between $(S,E)$ and $(%
\check{S};\check{E})$ is greater or equal to 1.

\item Or $E=\check{E}$. So the strategies $S$ and $\check{S}$ are not equal.
Let $n_{0}$ be the first index such that the terms $S$ and $\check{S}$
differ. Then
$
\forall k<n_{0},G_{f_{0}}^{n_{0}}(S,E)=G_{f_{0}}^{k}(\check{S},\check{E}),
$
and $G_{f_{0}}^{n_{0}}(S,E)\neq G_{f_{0}}^{n_{0}}(\check{S},\check{E})$.

As $E=\check{E},$ the cell that has changed in $E$ at the $n_{0}$-th
iterate is not the same than the cell that has changed in $\check{E}$, so
the distance between $G_{f_{0}}^{n_{0}}(S,E)$ and $G_{f_{0}}^{n_{0}}(\check{S%
},\check{E})$ is greater or equal to 2.
\end{itemize}
So the expansivity property is established.
\end{proof}

\begin{remark}
Expansivity is a kind of avalanche effect: any initial error is always magnified when iterating the system.
\end{remark}

\begin{remark}
$(\mathcal{X},G_{f_0})$ is not $A$-expansive, for any $A > 1$: let us consider two points $X=(E,S)$ and $X'=(E',S')$ with the same strategy ($S=S'$) and only one different cell ($d_e(E,E')=1$). So, $\forall n \in \mathds{N}, d_e\left(\mathcal{E}\left(G_{f_0}^n(X)\right),\mathcal{E}\left(G_{f_0}^n(X')\right)\right) = 1$.
\end{remark}

\subsection{Topological Entropy}

Another important tool to measure the chaotic behavior of a dynamical system is the topological entropy, which is defined for compact topological spaces.
Before studying the entropy of CIs, we must then check that $(\mathcal{X},d)$ is compact.

\subsubsection{Compacity Study}

In this section, we will prove that $(\mathcal{X},d)$ is a compact topological space, in order to study its topological entropy later.
Firstly, as $(\mathcal{X},d)$ is a metric space, it is separated. It is however possible to give a direct proof of this result:

\begin{theorem}
$(\mathcal{X}, d)$ is a separated space.
\end{theorem}

\begin{proof}
Let $(E,S) \neq (\textrm{\^{E}},\textrm{\^{S}})$ two points of $\mathcal{X}$.

\begin{enumerate}
\item If $E \neq \textrm{\^{E}}$, then the intersection between the two balls  $\mathcal{B}\left((E,S),\frac{1}{2}\right)$ and $\mathcal{B}\left((\textrm{\^{E}},\textrm{\^{S}}), \frac{1}{2}\right)$  in empty.
\item Else, it exists $k\in\mathds{N}$ such that $S^k \neq \textrm{\^{S}}^k$, then the balls $\mathcal{B}\left((E,S),10^{-(k+1)}\right)$ and $\mathcal{B}\left((\textrm{\^{E}},\textrm{\^{S}}), 10^{-(k+1)}\right)$ can be chosen.
\end{enumerate}
\end{proof}

The sequential characterization of the compacity for metric spaces can now be used to obtain the following result.

\begin{theorem}
$(\mathcal{X},d)$ is a compact metric space.
\end{theorem}

\begin{proof}
Let $(E^n,S^n)_{n \in \mathds{N}}$ be a sequence of $\mathcal{X}$.
\begin{enumerate}
\item A state $E^{\textrm{\~{n}}}$ which appears an infinite number of time in this sequence can be found. Let
$$I = \{ (E^n, S^n) \big/ E^n=E^{\textrm{\~{n}}}\}.$$
For all $(E,S) \in I$, $S_0^n \in \llbracket 1, \mathsf{N} \rrbracket$, and $I$ is an infinite set. Then $\textrm{\~{k}} \in \llbracket 1, \mathsf{N} \rrbracket$ can be found such that an infinite number of strategies of $I$ starts with $\textrm{\~{k}}$.

Let $n_0$ be the smallest integer such that $E^n = E^{\textrm{\~{n}}}$ and $S_0^n = \textrm{\~{k}}$.

\item The set
$$I' = \{(E^n,S^n) \big/ E^n = E^{n_0} \textrm{ and } S_0^n = S_0^{n_0}\}$$

is infinite, then one of the elements of $\llbracket 1, \mathsf{N} \rrbracket$ will appear an infinite number of times in the $S_1^n$ of $I'$: let us call it $\textrm{\~{l}}$.

Let $n_1$ be the smallest $n$ such that $(E^n,S^n) \in I'$ and $S_1^n = \textrm{\~{l}}$.

\item The set
$$I'' = \{(E^n,S^n) | E^n = E^{n_0} \textrm{ and } S_0^n = S_0^{n_0} \textrm{ and } S_1^n = S_1^{n_1}\}$$

is infinite, \emph{etc.}
\end{enumerate}

\noindent Let $l = \left(E^{n_0},\left(S_k^{n_k}\right)_{k \in \mathds{N}}\right)$, then the subsequence $\left(E^{n_k},S^{n_k}\right)$ converges to $l$.
\end{proof}

\subsubsection{Topological entropy}

Let $(X, d)$ be a compact metric space and $f: X \rightarrow X$ be a continuous map. For each natural number $n$, a new metric $d_n$ is defined on $X$ by

$$d_n(x,y)=\max\{d(f^i(x),f^i(y)): 0\leq i<n\}.$$

Given any $\varepsilon > 0$ and $n \geqslant 1$, two points of $X$ are $\varepsilon$-close with respect to this metric if their first $n$ iterates are $\varepsilon$-close.

This metric allows one to distinguish in a neighborhood of an orbit the points that move away from each other during the iteration from the points that travel together. A subset $E$ of $X$ is said to be $(n, \varepsilon)$-separated if each pair of distinct points of $E$ is at least $\varepsilon$ apart in the metric $d_n$. Denote by $H(n, \varepsilon)$ the maximum cardinality of an $(n, \varepsilon)$-separated set,

\begin{definition}
The \emph{topological entropy} of the map f is defined by (see e.g.~\cite{Adler65} or~\cite{Bowen})
$$h(f)=\lim_{\epsilon\to 0} \left(\limsup_{n\to \infty} \frac{1}{n}\log H(n,\varepsilon)\right). $$
\end{definition}

We have the result,

\begin{theorem}
Entropy of $(\mathcal{X},G_f)$ is infinite.
\end{theorem}

\begin{proof}
Let $\textrm{E}, \textrm{\v{E}}\in \mathbb{B}^\mathsf{N}$ such that $\exists i_0 \in \llbracket 1, N \rrbracket, \textrm{E}_{i_0} \neq \textrm{\v{E}}_{i_0}$. Then, $\forall \textrm{S}, \textrm{\v{S}} \in \mathcal{S}$,
$$d((\textrm{E},\textrm{S});(\textrm{\v{E}},\textrm{\v{S}})) \geqslant 1$$
But the cardinal $c$ of $\mathcal{S}$ is infinite, then $\forall n \in \mathbb{N}, c >e^{n^2}$.

Then for all $n \in \mathbb{N}$, the maximal number $H(n,1)$ of $(n,1)-$separated points is greater than or equal to $e^{n^2}$, so
$$h_{top}(G_f,1) = \overline{lim} \frac{1}{n} log \left( H(n,1)\right) > \overline{lim} \frac{1}{n} log \left( e^{n^2} \right) = \overline{lim} ~(n) = + \infty.$$

\noindent But $h_{top}(G_f,\varepsilon)$ is an increasing function when  $\varepsilon$ is decreasing, then

$$h_{top} \left( G_f \right) = \lim_{\varepsilon \rightarrow 0} h_{top}(G_f,\varepsilon) > h_{top}(G_f,1) = + \infty,$$
\noindent which concludes the evaluation of the topological entropy of $G_f$.
\end{proof}

We have proven that it is possible to find $f$, such that chaotic
iterations generated by $f$ can be described by a chaotic and entropic map on a
topological space in the sense of Devaney. We have considered a finite set
of states  $\mathds{B}^{\mathsf{N}}$ and a set $\mathbb{S}$ of strategies
composed by an infinite number of infinite sequences. In the following
section we will discuss the impact of these assumptions in the context of the
finite set of machine numbers.

\section{CHAOS IN A FINITE STATE MACHINE}
\label{section:CHAOS IN A FINITE STATE MACHINE}

Let us now explain how it is possible to have a true chaos in a finite state machine.

\subsection{A program with a chaotic behavior}

In the Section \ref{section:CHAOTIC ITERATIONS AS DEVANEY'S CHAOS} we have proven that discrete chaotic iterations can be put in the field of discrete dynamical systems:
$$
\left\{
\begin{array}{l}
x^{0}\in \mathcal{X} \\
x^{n+1} = G_f(x^{n}),
\end{array}%
\right.
$$
where $(\mathcal{X},d)$ is a metric space and $G_f$ is a continuous function. Thus, it becomes possible to study the topological behavior of those chaotic iterations.
Precisely, it has been proven that if the iterate function is based on the vectorial logical negation $f_0$, then chaotic iterations generate chaos according to Devaney. Therefore chaotic iterations, as Devaney's topological chaos, satisfy: sensitive dependence on the initial conditions, unpredictability, indecomposability, and uniform repartition.
Additionally, $G_{f_0}$ has been proven to be expansive and topologically mixing, and its topological entropy has been computed.
Our intention is now to use these chaotic iterations that are highly unpredictable, to build programs in the computer science security field.
Furthermore, we will give in Section \ref{sec:A CHAOTIC NEURAL NETWORK AS HASH FUNCTION} a link between CIs and artificial neural networks, thus it is possible to make them behave chaotically.

\medskip

Up to now, most of computer programs presented as chaotic lose their chaotic properties while computing in the finite set of machine numbers.
The algorithms that have been presented as chaotic usually act as follows.
After having received its initial state, the machine works alone with no interaction with the outside world.
Its outputs only depend on the different states of the machine.
The main problem which prevents speaking about chaos in this particular situation is that when a finite state machine reaches a same internal state twice, the two future evolutions are identical.
Such a machine always finishes by entering into a cycle while iterating.
This highly predictable behavior cannot be set as chaotic, at least as expressed by Devaney.
Some attempts to define a discrete notion of chaos have been proposed, but they are not completely satisfactory and are less recognized than the notions exposed in this paper.

\medskip

The stated problem can be solved in the following way.
The computer must generate an output $O$ computed from its current state $E$ \emph{and} the current value of an input $S$, which changes at each iteration (Fig. \ref{fig:Mealy}).
Therefore, it is possible that the machine presents the same state twice, but with two future evolutions completely different, depending on the values of the input.
By doing so, we thus obtain a machine with a finite number of states, which can evolve in infinitely different ways, due to the new values provided by the input at each iteration.
Thus such a machine can behave chaotically, as defined in the Devaney's formulation.

\begin{figure}[h!]
\centerline{\includegraphics[scale=0.5]{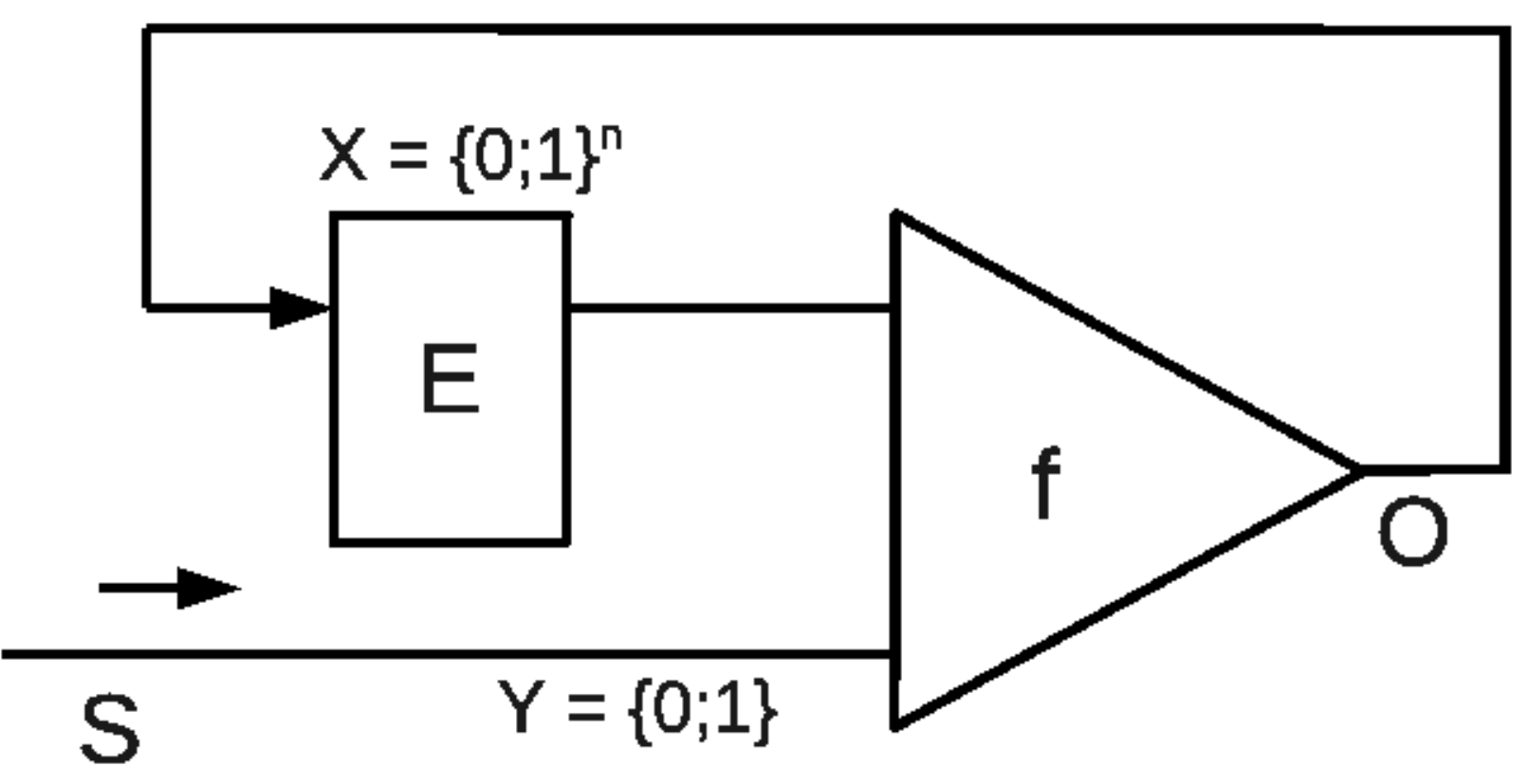}}
\caption{A chaotic finite-state machine. At each iteration, a new value is taken from the outside world (S). It is used by f as input together with the current state (E).}
\label{fig:Mealy}
\end{figure}

\subsection{The practical case of finite strategies}
\label{subsection:The particular case of regularity}

It is worthwhile to notice that even if the set of machine numbers is finite, we deal in practice with the \emph{infinite} set of strategies that have \emph{finite but unbounded lengths}.
Indeed, it is not necessary to store all of the terms of these strategies in the memory.
Only its $n^{th}$ term (an integer less than or equal to $\mathsf{N}$) has to be stored at the $n^{th}$ step, as it is illustrated in the following example.
Let us suppose that a given text is input from the outside world into the computer character by character and that the current term of the strategy is computed from the ASCII code of the current stored character. Since the set of all possible texts of the outside world is infinite and the number of their characters is unbounded, we work with an infinite set of finite but unbounded strategies.

In the computer science framework, we also have to deal with a finite set of states of the form $\mathds{B}^\mathsf{N}$ and as stated before an infinite set $\mathbb{S}$ of strategies. The sole difference with the theoretical study is that instead of being infinite the sequences of $S$ are finite with unbounded length, as any reasonable program must obviously finish one day.

The proofs of continuity and transitivity stated previously are independent of the finiteness of the length of strategies (sequences of $\mathbb{S}$). Sensitivity can be proven too in this situation (see Section~\ref{par:Sensitivity}). So even in the case of finite machine numbers, we have the two fundamental properties of chaos: sensitivity and transitivity, which respectively implies unpredictability and indecomposability (see~\cite{Devaney}, p.50). The regularity supposes that the sequences are of infinite lengths. To obtain the analogous of regularity in the context of finite sets, we can for example define a notion of \emph{periodic but finite} sequences.

\begin{definition}
A strategy $S\in\mathbb{S}$ is said to be \emph{periodic but finite} if $S$ is a finite sequence of length $n$ and if there exists a divisor $p$ of $n$, $p \neq n$, such that $\forall i \leqslant n-p, S^i = S^{i+p}$. A point $(E,S) \in \mathcal{X}$ is said to be \emph{periodic but finite}, if its strategy $S$ is periodic but finite.
\end{definition}

In this situation, $(1,2,1,2,1,2,1,2)$ ($p$=2) and $(2,2,2)$ ($p$=1), are periodic but finite. This definition can be interpreted as the analogous of periodicity definition on finite strategies. Following the proof of regularity (Section \ref{regularite}), it can be proven that the set of periodic but finite points is dense on $\mathcal{X}$, hence obtaining a desired element of regularity in finite sets, as quoted by Devaney (\cite{Devaney}, p.50): ``two points arbitrary close to each other could have completely different behaviors, the one could have a cyclic behavior as long as the system iterates while the trajectory of the second could `visit' the whole phase space''. It should be recalled that the regularity was introduced by Devaney in order to counteract the effects of transitivity: two points close to each other can have fundamentally different behaviors.

\bigskip

In the following we explain how to use chaotic iterations in the computer science security field, by using an illustrative example.
In this sense, we show two different ways to compute chaotic hash functions, the second one using neural networks.

\section{HASH FUNCTIONS WITH TOPOLOGICAL CHAOS PROPERTIES}
\label{section:HASH FUNCTIONS BASED ON TOPOLOGICAL CHAOS}
In this section, a concrete example of a chaotic program is given in the computer science security field.

\subsection{Introduction}

The use of chaotic maps to generate hash algorithms has seen several developments in recent years. In \cite{Fei2005} for example, a digital signature algorithm based on an elliptic curve and chaotic mapping is proposed to strengthen the security of an elliptic curve digital signature algorithm. Other examples of the generation of a hash function using chaotic maps can be found in, \emph{e.g.}, \cite{Wang2003,Xiao20094346,Peng2005}.
Neural networks that have learned a continuous chaotic map have been proposed too in recent years \cite{Xiao10}, to achieve hash functions requirements.

Note that using any chaotic map does not guarantee that the resulting hash function would behave chaotically too.
To the best of our knowledge, this point is not discussed in these referenced papers, however it should be considered as important.
We define in this section a new way to construct hash functions based on chaotic iterations.
As a consequence of the theory presented before, the generated hash functions satisfy various topological chaos properties.
Thus, properties required for hash functions are guaranteed by our approach.
For example, the avalanche criterion is deduced from the expansivity property.

\subsection{A chaotic hash function}

In this section, we explain a new way to obtain a digest of a digital medium described by a binary sequence.
It is based on chaotic iterations and satisfies various topological chaos properties.
The hash value will be the last state of some chaotic iterations: the initial state $X_0$, finite strategy $S$, and iterate function must then be defined.
\label{subsubsec:initial}

The initial condition $X_0=\left( S,E\right) $ is composed by a $\mathsf{N} = 256$ bits sequence $E$ and a chaotic strategy $S$.
In the following section, we describe in detail how to obtain this initial condition from the original medium.

\subsubsection{How to obtain $E$}

The first step of our algorithm is to transform the message in a normalized 256 bits sequence $E$.
To illustrate this step inspired by SHA-1, we state that our original text is: ``\emph{The original text}''.
Each character of this string is replaced by its ASCII code (on 7 bits). Then, we add a 1 at the end of this string.
\bigskip
\begin{center}
\begin{alltt}
\noindent 10101001 10100011 00101010 00001101 11111100 10110100
\noindent 11100111 11010011 10111011 00001110 11000100 00011101
\noindent 00110010 11111000 11101001
\end{alltt}
\end{center}
\bigskip
So, the binary value (1111000) of the length of this string (120) is added, with another 1:
\bigskip
\begin{center}
\begin{alltt}
\noindent 10101001 10100011 00101010 00001101 11111100 10110100
\noindent 11100111 11010011 10111011 00001110 11000100 00011101
\noindent 00110010 11111000 11101001 11110001
\end{alltt}
\end{center}
\bigskip
This string is inverted (the last bit is now the first one) and the two new substrings are concatenated.
This gives:
\bigskip
\begin{center}
\begin{alltt}
\noindent 10101001 10100011 00101010 00001101 11111100 10110100
\noindent 11100111 11010011 10111011 00001110 11000100 00011101
\noindent 00110010 11111000 11101001 11110001 00011111 00101110
\noindent 00111110 10011001 01110000 01000110 11100001 10111011
\noindent 10010111 11001110 01011010 01111111 01100000 10101001
\noindent 10001011 0010101
\end{alltt}
\end{center}
\medskip So, we obtain a multiple of 512, by duplicating this string enough and truncating at the next multiple of 512. This string in which the whole original text is contained, is denoted by $D$.
\bigskip

Finally, we split the new string into blocks of 256 bits and apply the exclusive-or function, obtaining a 256 bits sequence in a manner inspired by the SHA-X algorithms.
\bigskip
\begin{alltt}
\noindent 11111010 11100101 01111110 00010110 00000101 11011101
\noindent 00101000 01110100 11001101 00010011 01001100 00100111
\noindent 01010111 00001001 00111010 00010011 00100001 01110010
\noindent 01000011 10101011 10010000 11001011 00100010 11001100
\noindent 10111000 01010010 11101110 10000001 10100001 11111010
\noindent 10011101 01111101
\end{alltt}

In the context of Subsection \ref{subsubsec:initial}, $\mathsf{N}=256$, and $E$ is the above obtained sequence of 256 bits: the given message has been compressed into a 256 binary string.

\medskip

We now have the definitive length of our digest.
Note that a lot of texts have the same normalized string.
This is not a problem because the strategy we will build depends on the whole text too, in such a way that two different texts lead to two different strategies.
Let us now build the strategy $S$.

\subsubsection{How to choose $S$}

To obtain the strategy $S$, an intermediate sequence $(u^n)$ is constructed from $D$ as follows:
\begin{itemize}
\item $D$ is split into blocks of 8 bits. Then $u^n$ is the decimal value of the $n^{th}$ block.
\item A circular rotation of one bit to the left is applied to $D$ (the first bit of $D$ is put on the end of $D$). Then the new string is split into blocks of 8 bits another time. The decimal values of those blocks are added to $(u^n)$.
\item This operation is repeated again 6 times.
\end{itemize}
\bigskip

It is now possible to build the strategy $S$:
\begin{equation*}
S^0 = u^0,~~~
S^n=(u^n+2\times S^{n-1}+n) ~(mod ~256).
\end{equation*}%
\noindent $S$ will be highly dependent to the changes of the original text, because $\theta \longmapsto 2\theta ~(mod ~1)$ is known to be chaotic as defined by Devaney's theory \cite{Devaney}.

\subsubsection{How to construct the digest}

To construct the digest, chaotic iterations are done with initial state $X^0$,
\begin{equation*}
\begin{array}{rccc}
f: & \llbracket1,256\rrbracket & \longrightarrow & \llbracket1,256\rrbracket
\\
& (E_1,\hdots,E_{256}) & \longmapsto & (\overline{E_1},\hdots,\overline{%
E_{256}}),%
\end{array}%
\end{equation*}%
\noindent as iterate function, and $S$ for the chaotic strategy.
\bigskip

\noindent The result of those iterations is a 256 bits vector. Its components are taken 4 bits at a time and translated into hexadecimal numbers, to obtain the hash value:
\medskip
\begin{small}
\begin{alltt}
\noindent 63A88CB6AF0B18E3BE828F9BDA4596A6A13DFE38440AB9557DA1C0C6B1EDBDBD
\end{alltt}
\end{small}

\bigskip

To compare, if instead of using the text \textquotedblleft \textit{The original text}\textquotedblright\ we took \textquotedblleft \textit{the original text}\textquotedblright , the hash function returns:
\medskip
\begin{small}
\begin{alltt}
\noindent 33E0DFB5BB1D88C924D2AF80B14FF5A7B1A3DEF9D0E831194BD814C8A3B948B3
\end{alltt}
\end{small}
\bigskip

In this paper, the generation of hash value is done with the vectorial Boolean negation $f_{0} $ defined in eq. (\ref{f0}). Nevertheless, the procedure remains general and can be applied with any function $f$ such that $G_f$ is chaotic. In the following subsection, a complete example of the procedure is given.

\subsection{Application example}

Consider two black and white images of size $64 \times 64$ in Fig. \ref{Hash of some black and white images}, in which the pixel in position (40,40) has been changed.
\begin{figure}[h]
\centering
\subfigure[Original image.]{\includegraphics[width=0.21\textwidth]{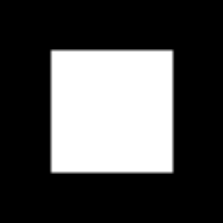}\label{Original image}}
\subfigure[Modified image.]{\includegraphics[width=0.21\textwidth]{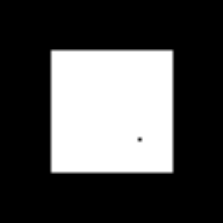}\label{Modified image}}
\caption{Hash of some black and white images.}
\label{Hash of some black and white images}
\end{figure}
\medskip
In this case, the hash function returns:
\begin{small}
\begin{alltt}
\noindent 34A5C1B3DFFCC8902F7B248C3ABEFE2C9C9538E5104D117B399C999F74CF1CAD
\end{alltt}
\end{small}
for the Fig. \ref{Original image} and
\begin{small}
\begin{alltt}
\noindent 5E67725CAA6B7B7434BE57F5F30F2D3D57056FA960B69052453CBC62D9267896
\end{alltt}
\end{small}
for the Fig. \ref{Modified image}.
\bigskip

Consider two 256 graylevel images of Lena ($256 \times 256$ pixels) in figure \ref{Hash of some grayscale level images}, in which the grayscale level of the pixel in position (50,50) has been transformed from 93 (fig. \ref{Original Lena}) to 94 (fig. \ref{Modified Lena}).
\begin{figure}[h]
\centering
\subfigure[Original lena.]{\includegraphics[width=0.38\textwidth]{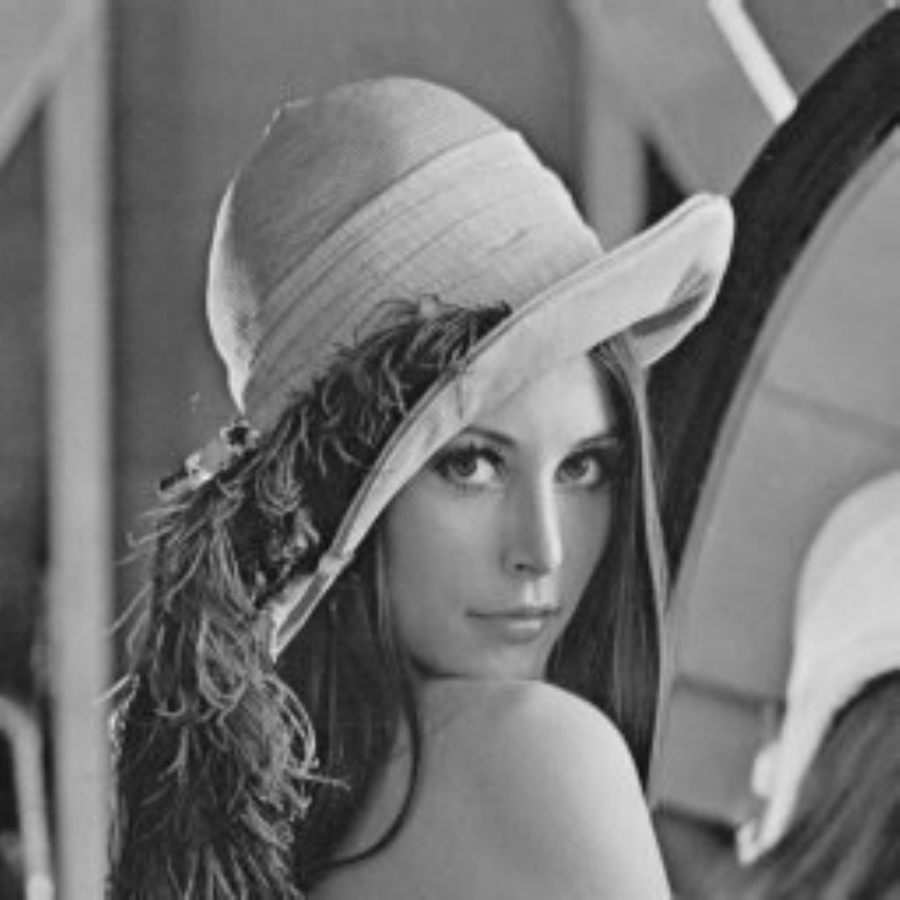}\label{Original Lena}}\hspace{1cm}
\subfigure[Modified lena.]{\includegraphics[width=0.38\textwidth]{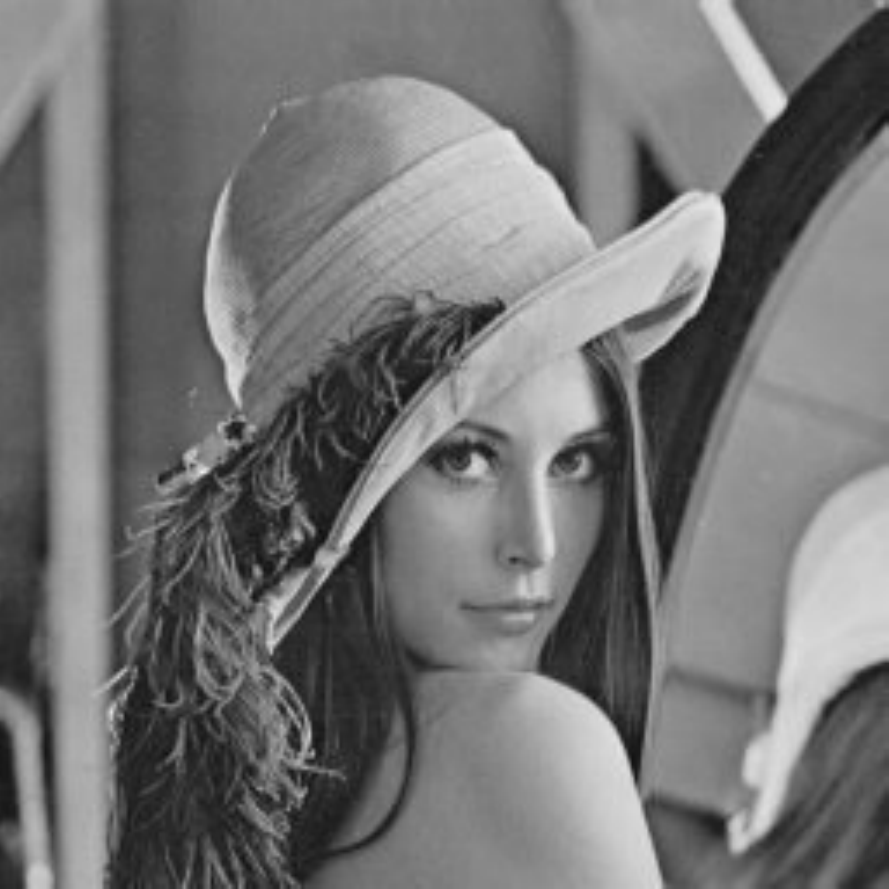}\label{Modified Lena}}
\caption{Hash of some grayscale level images.}
\label{Hash of some grayscale level images}
\end{figure}
In this case, the hash function returns:
\begin{small}
\begin{alltt}
\noindent FA9F51EFA97808CE6BFF5F9F662DCD738C25101FE9F7F427CD4E2B8D40331B89
\end{alltt}
\end{small}
for the left Lena and
\begin{small}
\begin{alltt}
\noindent BABF2CE1455CA28F7BA20F52DFBD24B76042DC572FCCA4351D264ACF4C2E108B
\end{alltt}
\end{small}
for the right Lena.
\medskip

These examples give an illustration of the avalanche effect obtained by this algorithm. A more complete study of the properties possessed by the hash functions and resistance under collisions will be studied in a future work.

\section{A CHAOTIC NEURAL NETWORK AS HASH FUNCTION}
\label{sec:A CHAOTIC NEURAL NETWORK AS HASH FUNCTION}

A hash function can be achieved in two stages: the compression of the message (mapping a binary sequence of any length $n \in \mathds{N}$ into a message of a fixed length belonging into $\mathds{B}^\mathsf{N}$, for a given fixed length $\mathsf{N} \in \mathds{N}$) and the hash of the compressed message \cite{Xiao10}.
As several compression functions have yet been proposed to achieve the first stage, we will only focus on the second stage and we will explain how to build a neural network that realize it.
This neural network that hashes compressed messages will behave chaotically, as it is defined by the Devaney's theory.

\medskip

Let us firstly explain how it is possible to build a neural network that behaves chaotically.
Consider $f:\mathds{B}^\mathsf{N} \longrightarrow \mathds{N}^\mathsf{N}$ and a MLP which recognize $F_{f}$.
That means, for all $(k,x) \in \llbracket 1 ; \mathsf{N} \rrbracket \times
\mathds{B}^\mathsf{N}$, the response of the output layer to the input
$(k,x)$ is $F_{f}(k,x)$.
We thus connect the output layer to the input one as it is depicted in
Figure~\ref{perceptron}, leading to a global recurrent artificial neural network (ANN) working as follows \cite{arxivRNNchaos}:

\begin{figure}[!t]
\centering
\includegraphics[width=3.5in]{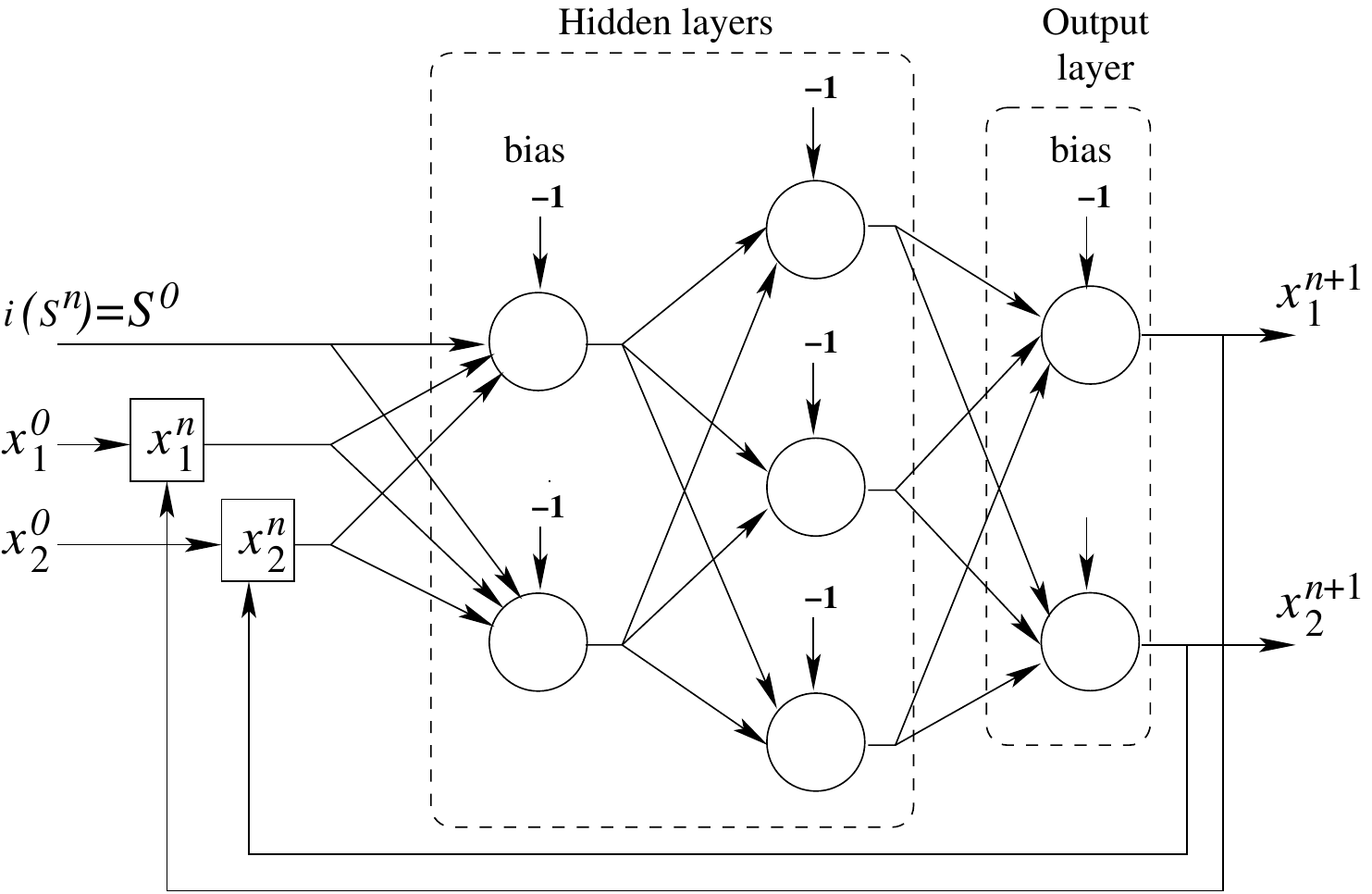}
\caption{Example of global recurrent neural network modeling function $F_{f}$ such that
 $x^{n+1}=\left(x^{n+1}_1,x^{n+1}_2\right)=F_{f}\left(i(S^n),\left(x^n_1,x^n_2\right)\right)$}
\label{perceptron}
\hfil
\end{figure}

\begin{itemize}
\item At the initialization stage, the ANN receives a Boolean vector
$x^0\in\mathds{B}^\mathsf{N}$ as input state, and $S^0 \in \llbracket
1;\mathsf{N}\rrbracket$ in its input integer channel $i()$. Thus,
$x^1 = F_{f}(S^0, x^0)\in\mathds{B}^\mathsf{N}$ is computed by the
neural network.
\item This state $x^1$ is published as an output. Additionally, $x^1$
is sent back to the input layer, to act as Boolean state in the next
iteration.
\item At iteration number $n$, the recurrent neural network receives
the state $x^n\in\mathds{B}^\mathsf{N}$ from its output layer and
$i\left(S^n\right) \in \llbracket 1;\mathsf{N}\rrbracket$ from its
input integer channel $i()$. It can thus calculate $x^{n+1} =
F_{f}(i\left(S^n\right), x^n)\in\mathds{B}^\mathsf{N}$, which will
be the new output of the network.
\end{itemize}

Obviously, this particular MLP produce exactly the same values than CIs with update function $f$. That is, such MLPs are equivalent to CIs with $f$ as update function.
However, the compression stage of the hash function presented in the previous section can be resumed to making chaotic iterations over the compressed message.
As chaotic iterations can be obtained with a neural network, we can thus realize this stage with a (chaotic) neural network.
Finally, it is important to remark that the proposed hash function can be implemented into a global neural network, as various compression neural networks can be found in the literature \cite{Mahoney00,Rudenko08,Cramer96}: we just have to replace our compression stage, inspired by SHA-X, with one of these compression ANN.

\section{CONCLUSION}

In this paper, a new approach to compute programs with a chaotic behavior is proposed.
This approach is based on the well-known Devaney's topological chaos.
The algorithms which are of iterative nature are based on the so-called chaotic iterations.
This is achieved by establishing a link between the notions of topological chaos and chaotic iterations.
Indeed, we are not interested in stable states of such iterations as it has always been the case in the literature, but in their unpredictable behavior.
After a solid theoretical study, we consider the practical implementation of the proposed algorithms by evaluating the case of finite sets. We study the behavior of the induced computer programs proving that it is possible to design true chaotic computer programs.

An application is proposed in the area of information security: a new hash function is presented, the security in this case is guaranteed by the unpredictability of the behavior of the proposed algorithms.
 The algorithms derived from our approach satisfy important properties of topological chaos such as sensitivity to initial conditions, uniform repartition (as a result of the transitivity), unpredictability, expansivity, and topological mixing.
 Moreover, its topological entropy is infinite.
The results expected in our study have been experimentally checked. The choices made in this first study are simple: compression function inspired by SHA-1, negation function for the iteration function, \emph{etc.}
The aim was not to find the best hash function, but to give simple illustrated examples to prove the feasibility in using the new kind of chaotic algorithms in computer science.
Finally, we have shown how the mathematical framework of topological chaos offers interesting qualitative and qualitative tools to study the algorithms based on our approach.

In future work, we will investigate other choices of iteration functions and chaotic strategies. We will try to characterize transitive functions.  Other properties induced by topological chaos will be explored and their interest in the information security framework will be deepened.

\bibliographystyle{plain}
\bibliography{mabase}
\end{document}